\documentclass[aps,prl,twocolumn,superscriptaddress,longbibliography,10pt]{revtex4-1}
\usepackage{amsmath,amsthm,amssymb,amsfonts}
\usepackage{graphicx}
\usepackage{hyperref}
\usepackage{dsfont}
\usepackage{xcolor}
\usepackage{graphicx}
\usepackage{enumitem}
\usepackage{mathrsfs}
\newcommand{\Tr}[1]{\operatorname{Tr}\!\left[#1\right]}
\newcommand{\PTr}[2]{\operatorname{Tr}_{\mathrm{#1}}\!\left[#2\right]}
\newcommand{\PTrOp}[1]{\operatorname{Tr}_{\mathrm{#1}}}
\def\>{\rangle}
\def\<{\langle}

\def\id{\mathsf{id}}

\def\mD{\mathcal{D}}
\def\mN{\mathcal{N}}

\def\from{\leftarrow}
\newcommand{\ketbra}[2]{\left|{#1}\right>\!\!\left<{#2}\right|}

\newcommand{\set}[1]{\mathscr{#1}}

\newcommand{\defeq}{\stackrel{\textup{\tiny def}}{=}}

\renewcommand{\ge}{\geqslant}
\renewcommand{\le}{\leqslant}

\newtheorem{theorem}{Theorem}
\newtheorem{corollary}{Corollary}

\newtheorem{proposition}{Proposition}

\newtheorem{definition}{Definition}

\theoremstyle{remark}

\def\payoff{\wp}

\def\openone{\mathds{1}}
\def\free{\mathcal{F}}
\def\freeop{\mathcal{L}}
\def\sufficient{\rightarrowtail}
\def\ssscenario{\mathsf{S}}

\newcommand{\mI}{\mathcal{I}}
\newcommand{\EB}{\operatorname{EB}}
\newcommand{\mU}{\mathcal{U}}

\newcommand{\on}[1]{^\mathrm{#1}}
\newcommand{\hilbert}[1]{\ensuremath{\mathsf{H}_\mathrm{#1}}}
\newcommand{\density}[1]{\ensuremath{\mathsf{D}(\hilbert{#1})}}
\newcommand{\subdensity}[1]{\ensuremath{\mathsf{D}_\le(\hilbert{#1})}}
\newcommand{\densitytwo}[2]{\ensuremath{\mathsf{D}_{\mathrm{#1}\otimes\mathrm{#2}}}}
\newcommand{\channel}[1]{\ensuremath{\mathcal{#1}}}
\usepackage{soul,color}

\newcommand{\sigscenario}[1]{\ssscenario^\text{sig}_{#1}}
\newcommand{\sigcorrelation}[1]{p^\text{sig}_{#1}}
\newcommand{\inputset}[1]{\mathcal{Q}_\text{#1}}
\newcommand{\choi}[1]{\mathsf{J}_{#1}}

\newcommand{\vacuum}{\ketbra{\varnothing}{\varnothing}}

\begin{document}

%\SetWatermarkText{NOT FOR DISTRIBUTION}
%\SetWatermarkAngle{60}
%\SetWatermarkScale{0.5}

\title{A resource theory of quantum memories and their faithful verification with minimal assumptions}
\author{Denis Rosset}
\email{physics@denisrosset.com}
\affiliation{Department of Physics, National Cheng Kung University, Tainan 701, Taiwan}
\affiliation{Perimeter Institute for Theoretical Physics, Waterloo, Ontario, Canada, N2L 2Y5}
\author{Francesco Buscemi}
\email{buscemi@i.nagoya-u.ac.jp}
\affiliation{Department of Mathematical Informatics, Nagoya University, Chikusa-ku, Nagoya, 464-8601, Japan}
\author{Yeong-Cherng Liang}
\email{ycliang@mail.ncku.edu.tw}
\affiliation{Department of Physics, National Cheng Kung University, Tainan 701, Taiwan}

\date{\today}

\begin{abstract}
  We provide a complete set of game-theoretic conditions equivalent to the existence of a transformation from one quantum channel into another one, by means of classically correlated pre/post processing maps only.
  Such conditions naturally induce tests to certify that a quantum memory is capable of storing quantum information, as opposed to memories that can be simulated by measurement and state preparation (corresponding to entanglement-breaking channels).
  These results are formulated as a resource theory of genuine quantum memories (correlated in time), mirroring the resource theory of entanglement in quantum states (correlated spatially).
  As the set of conditions is complete, the corresponding tests are faithful, in the sense that any non entanglement-breaking channel can be certified.
  Moreover, they only require the assumption of trusted inputs, known to be unavoidable for quantum channel verification.
  As such, the tests we propose are intrinsically different from the usual process tomography, for which the probes of both the input and the output of the channel must be trusted.
  An explicit construction is provided and shown to be experimentally realizable, even in the presence of arbitrarily strong losses in the memory or detectors.
\end{abstract}

\maketitle

%%%%%%%%%%%%%%%%%%%%%%%%%%%%%%%%%%%%%%%%%%%

Consider a vendor selling quantum devices purportedly able of storing quantum information for a period of time.
However, during their operation, the devices always break the entanglement between the stored subsystem and any other subsystem.
Such devices are arguably useless as quantum memories; for example, they would not be able to establish entangled links in a quantum repeater scheme~\cite{Briegel1998,Sangouard2011}.
In the terminology of quantum channels~\cite{Holevo1998,Braunstein1998,Verstraete2002a}, those devices correspond to entanglement-breaking (EB) channels~\cite{Horodecki2003, Holevo2008}, which are exactly equivalent to the measure-and-prepare channels depicted in Fig.~\ref{fig:eb-processors}.
Measure-and-prepare channels are implemented by measuring the input state, storing the classical information corresponding to the measurement outcome for the required duration, and then using that information to prepare a quantum state.
Even though strictly speaking, such channels are quantum channels (since they act upon an input quantum system), they actually transmit only classical information from the input to the output.
Thus, in constructing a quantum memory, one aims to produce a device that could retain some correlations between the stored system and remote systems.

\begin{figure}
  \includegraphics{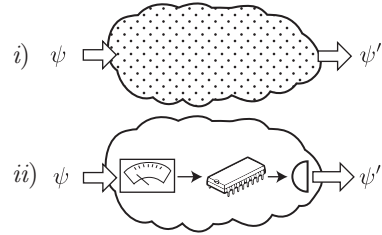}
  \caption{
    \label{fig:eb-processors}
    The quantum device in i) is {\em not} in the quantum domain whenever its functioning can be simulated by a fixed measurement, an arbitrary processing of the classical outcome, and a preparation, as shown in ii).
    Such devices are called \textit{entanglement-breaking (EB) channels}.
  }
\end{figure}

Due to its relevance in quantum information science, the benchmarking of quantum memories has been extensively considered in the literature~\cite{Namiki2008,Haseler2008,Haseler2009,Macchiavello2013,Pusey2015}.
An obvious way to verify whether a channel is EB or not is by performing process tomography~\cite{Poyatos1997,DAriano2001,Mohseni2008}, that is, by feeding through the channel a tomographically complete set of states and performing a tomographically complete measurement at the output.
By collecting sufficient statistics, it is possible to reconstruct the process matrix corresponding to the channel up to any desired level of accuracy, and check whether the channel is EB or not.
This scheme, however, requires complete control on all parts of the experiment: the state preparation device and measurement device must be both accurate and trusted, i.e., the experimenter must trust that they are doing \textit{exactly} what they would like to achieve (or, at least, what they think they are doing).

The possibility of lifting the assumption of trust in whole (or parts) of the experimental setup is at the roots of an important area of research known as (partially)-\textit{device-independent} quantum information~\cite{Scarani2012,Rosset2012a,Brunner2014}, with important applications not only to quantum technologies~\cite{Ekert1991,Acin2007,Vazirani2014,Colbeck2006,Pironio2010,Colbeck2011,Reichardt2013}, but also to the foundations of quantum theory~\cite{Bell1964,Bancal2012a,DallArno2017a,DallArno2017b}.
Indeed, a fully device-independent approach to quantum channel verification can be achieved in practice (see Fig.~\ref{fig:bell-test}).
An initial entangled state $\varphi\on{RA}$ is prepared and subsystem A is fed through the channel to be tested, denoted by $\channel{N}$, with its output subsystem denoted by B.
Channel $\channel{N}$ may represent the actual quantum communication channel delivering subsystem A from one location to another, or perhaps just a static quantum memory, storing subsystem A while R is moved to a different location.
However, both these scenarios are mathematically equivalent, as the final result, in both cases, is that the initial bipartite state $\varphi\on{RA}$ is transformed into another bipartite state $\tilde{\varphi}\on{RB}=(\id\on{R}\otimes\mN\on{A})\varphi\on{RA}$, now shared between two different locations.
At this point, a Bell test can be performed.
The violation of any Bell inequality between R and B would constitute an irrefutable evidence that the state $\tilde{\varphi}\on{RB}$ is entangled, and hence that the channel $\mN$ is in the quantum domain, even in the case in which both preparation and measurement devices are untrusted.

\begin{figure}
  \includegraphics{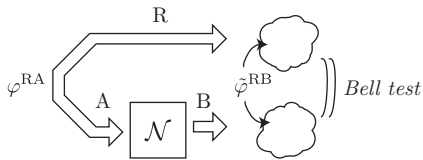}
  \caption{Quantumness verification via Bell tests.}
  \label{fig:bell-test}
\end{figure}

This approach allows one to verify a quantum channel in a fully device-independent way, namely, based solely on the correlations observed in the experiment.
It has however some limitations:
\begin{enumerate}
\item The ability to prepare entangled states is assumed.
\item Subsystem R must be preserved \textit{intact} during the waiting phase.
  Of course one could imagine to store $R$ in another copy of the same memory $\channel{N}$, but we would then be testing the quantumness of $\channel{N}^{\otimes 2}$, which is generally more demanding than testing $\channel{N}$ alone.
\item The violation of a Bell inequality is a sufficient condition for quantum entanglement, but it is known \textit{not to be necessary}, as there exist entangled states that do not violate any Bell inequality~\cite{Werner1989,Barrett2002}.
  In our framework, this is equivalent to saying that there exist genuinely quantum channels that would never pass this test~\cite{Pal2015}.
\item Bell tests are very fragile with respect to losses (during storage or detection)~\cite{Santos1992,Brunner2014}.
\end{enumerate}
Experimental violations of the Clauser-Horne-Shimony-Holt inequality (CHSH)~\cite{Clauser1969} have been reported for quantum memories~\cite{Ding2015,Ding2015a,Tiranov2015}, but without closing all loopholes (in particular the detection and locality loopholes).

Ideally then, we would like to lift all the above four assumptions and construct tests that:
\begin{enumerate}
\item do not require entangled states;
\item do not require the use of any additional side-channel (notice that an identity channel, as it appears for example in Fig.~\ref{fig:bell-test}, counts as one);
\item are faithful, i.e., are able to verify any channel in the quantum domain;
\item are loss tolerant;
\item trust neither the preparation of the input nor the measurement of the output.
\end{enumerate}
Alas, the above desiderata cannot be all met. In particular, the intimately related conditions~(1) and~(2) force us to test the channel in a time-like setting, thus ruling out, not only in practice but also \textit{in principle}, a fully device-independent solution.
In this paper, following Ref.~\cite{Pusey2015}, we give up the possibility to mistrust the preparation device, that is, we assume that the skeptic who wants to be convinced of the quantumness of the memory is able to prepare trusted states.

The necessity of trusting the input to the device is a consequence of the causal configuration of a channel test, in which the output of the channel necessarily lies in the future of the input.
In such a configuration, since we must allow any amount of classical communication for free (as we are testing for quantumness), if the input to the device is untrusted, then \textit{any} correlation can be reproduced, and no conclusion can be drawn about the nature of the channel.
In other words, the trusted input assumption is \textit{minimal} in quantum channel verification\footnote{This is to be compared with the protocol of \textit{quantum channel falsification}, as considered in~\cite{DallArno2017a,DallArno2017b}. There, no side channel, not even a classical one, is freely granted, but must be part of the initial hypothesis (the one to be falsified). As a consequence, channel falsification can be done in a fully device-independent way, whereas quantum channel verification (the task considered in this paper) can \textit{at most} be made measurement device-independent, but not fully device-independent.}.

We will thus work within the so-called measurement device-independent (MDI) framework~\cite{Branciard2013}, in particular taking inspiration from semiquantum nonlocal games~\cite{Buscemi2012}, which generalize the usual Bell-local games by enabling the referee to send quantum input states to the players.
Until now, this framework has only been applied to scenarios involving space-like separated systems, where all entangled states can be detected~\cite{Buscemi2012,Branciard2013}; this characterization is faithful even in the presence of arbitrary losses and classical communication between the systems~\cite{Rosset2013a}.
Other instances of its application include steering scenarios~\cite{Cavalcanti2013a} and nonclassical teleportation~\cite{Cavalcanti2017a}.
Some of these tests have been implemented experimentally~\cite{Xu2014,Nawareg2015,Kocsis2015,Verbanis2016}; and even when no prior knowledge about the tested systems is available, the verification can be performed directly on experimental data~\cite{Yuan2016,Verbanis2016,Rosset2017b}.
Within this framework, several authors proposed structural or quantitative tests of entanglement~\cite{Zhao2016,Shahandeh2017,Supic2017a,Rosset2017b} or the quantification of generated randomness~\cite{Cao2015,Chaturvedi2015,Supic2017a}.
In our construction below, we apply the MDI framework to temporal correlations arising out of the use of a quantum memory.

\section{The setup}

The construction that we propose here is best formulated as a game played between a referee (the buyer) and an experimentalist (the vendor) trying to convince the referee about her ability to store quantum information.
Contrary to Bell games which involve multiple players, our temporal game involve only one player, inquired at two successive instants in time.
It is then reasonable to assume that the player can retain an arbitrarily large amount of classical information (her memory) during the game.
This is also the reason why we explicitly avoid here speaking of two different players, one receiving the first question (Alice) and another one receiving the second question (Bob) --- nothing prevents Bob from being just Alice in the future.
In the present setup of ``one player at two different times'',  we name this player ``Abby'' and impose the same operational assumptions for her at all times. A schematic representation is given in Fig.~\ref{fig:game-config}.

\begin{figure}
  \includegraphics{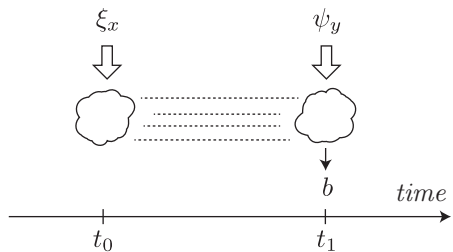}
  \caption{
    \label{fig:game-config}
    The test configuration.
    Time flows from left to right.
    At time $t=t_0$ a first input $\xi_x$ (that must be trusted by the referee who prepares it) is handed over to the experimentalist (whose laboratory is represented by the flying cloud).
    After some time, at $t=t_1>t_0$, the referee hands over a second input $\psi_y$.
    At this point the experimentalist is required to output a classical outcome, labeled by $b$.
    By repeating this procedure many times and observing the input/output correlation $p(b|x,y)$ so obtained, it is possible to verify whether the experimentalist is actually able or not to preserve quantum information through a time $\delta=t_1-t_0>0$.
  }
\end{figure}

More explicitly, the situation considered here is as follows:
\begin{enumerate}
\item At time $t=t_0$, the referee provides a first quantum input to Abby.
  Abby knows the set from which the state has been chosen, but she does not know the actual state given to her, while the referee has perfect knowledge of the corresponding density matrix $\xi_x$.
\item The referee then waits for some time, necessary for Abby to implement the specific channel or memory.
  Then, at time $t=t_1=t_0+\delta>t_0$, the referee hands a second quantum state $\psi_y$ over to Abby.
  Again, Abby knows the set of possible states, but is ignorant about the actual state that she receives.
\item After Abby receives $\psi_y$, she is required to broadcast a classical outcome, labeled by $b$.
\item By repeating the same procedure a sufficient number of times, the referee, by looking at the input/output correlation $p(b|x,y)$ obtained, can decide whether Abby has been able to store quantum information during a time $\delta$, or not.
\end{enumerate}

A comparison with the framework of Ref.~\cite{Pusey2015} is in order.
There also, the input to the channel is trusted (the referee can trust the state preparation device), however, the second inquiry is restricted to be a classical question.
This is the reason why the scheme proposed in Ref.~\cite{Pusey2015} is not able to detect all channels in the quantum domain, but only those corresponding to a steerable Choi operator.
Here, instead, we take the second question to be again encoded on a trusted quantum state.
Notice that this does not constitute an extra assumption: if the referee must trust the input she prepares at time $t_0$ (otherwise quantum channel verification is impossible), there is no reason to mistrust the preparation device at time $t_1$ (just to trust it again at the next round of the game).
Indeed, since the game is assumed to go on until sufficient statistics is collected, the assumption of a ``reusable'' trusted preparation device is essential.
Notice also that we \textit{never} require the ability to prepare entangled or classically correlated states: the quantum questions can always be drawn independently, so the operational assumptions about the preparation device are exactly those in Ref.~\cite{Pusey2015}, namely, of a trusted, re-usable, re-initializable (in information theory jargon, i.i.d.) quantum state preparation device.
The same i.i.d. assumption is made about the quantum memory being tested.

The rest of the paper is organized as follows: in Section~\ref{sec:resourcetheory} we put our discussion on formal ground, by formulating it as a resource theory of quantum memories, including a definition of monotones ranking the usefulness of memories.
In Section~\ref{sec:existence}, we show that the protocol described above defines a family of monotones able to order precisely all resources.
In particular, they enable the verification of the quantumness of any channel as soon as it is not entanglement-breaking.
The above protocol hence provides \textit{faithful} tests (i.e. necessary and sufficient), thus improving on those in Ref.~\cite{Pusey2015} which are only sufficient.
In Section~\ref{sec:construction} we will complement the existence proof with a way to constructively find practical tests, and discuss the role of losses in practical implementations.
In particular, we argue that the present approach is robust against typical models of loss, contrarily to what happens in Bell tests.

\section{Quantum memories as resources}
\label{sec:resourcetheory}
Let us consider a quantum channel $\channel{N}$ between finite Hilbert spaces $\hilbert{A}$ and $\hilbert{B}$, governing the evolution of an arbitrary state stored inside a quantum memory.
The channel $\channel{N}$ is a completely positive trace-preserving (CPTP) linear map $\channel{N}: \density{A} \to \density{B}$.
We write $\density{A}$ and $\density{B}$ the sets of density matrices respectively for $\hilbert{A}$ and $\hilbert{B}$; later, we will  write $\subdensity{A}$ for subnormalized density matrices.
The original version of semiquantum games~\cite{Buscemi2012} was constructed to detect the presence of entanglement in quantum states.
The present work, however, concerns the detection of channels in the quantum domain.
While resource theories have been extensively studied for quantum states~\cite{Horodecki2012,Brandao2015,Coecke2016,Streltsov2017,Buscemi2017,Buscemi2017} (correlations in space), limited results exist for quantum channels~\cite{BenDana2017} (correlations in time).
We thus introduce below a resource theory of quantum channels that mirrors the resource theory of entanglement for quantum states.

We include the following free operations in our resource theory.
First, we allow local operations, which are manipulations of a quantum system at a specific point in spacetime.
Second, we allow unlimited storage of classical information; this includes the use of preexisting randomness.
However, we do not include storage of quantum information (that is when it cannot be otherwise simulated with the free operations above).

For arbitrary (finite dimensional) $\hilbert{A}$ and $\hilbert{B}$, we consider the set $\free$ of quantum channels storing only classical information, the so-called measure-and-prepare or quantum-classical-quantum channels:
\begin{equation}
  \label{eq:entanglementbreaking}
  \channel{N}\on{A\to B}(\rho\on{A}) = \sum_\mu \pi(\mu) \sum_i {\rho'}_{i,\mu}\on{B} \Tr{\Pi_{i|\mu}\on{A} \rho\on{A}}
\end{equation}
where $\mu$ is preexisting randomness distributed according to $\pi(\mu)$; the measurement on the channel input $\hilbert{A}$ is described by the positive-operator valued measure (POVM) elements $\{ \Pi_{i|\mu}\on{A} \}$ that possibly depend on $\mu$; the family $\{{\rho'}_{i,\mu}\on{B} \}$ of density matrices prepared at the channel output $\hilbert{B}$ is indexed by the previously stored measurement result $i$ and the classical index $\mu$ (see Fig.~\ref{fig:ebprocessors}ii).
For ease of presentation when there is no ambiguity, we shall sometimes suppress the superscripts used to indicate the Hilbert space on which each operator acts.

Channels that can be written as in~(\ref{eq:entanglementbreaking}) are also known as entanglement-breaking (EB), as they are exactly those channels that, when applied locally to any bipartite state, produce a separable output~\cite{Horodecki2003}. We notice that the usual definition of EB channels does not explicitly account for the mixing index $\mu$, however, it is a straightforward matter to verify that the two are equivalent (see Appendix~\ref{app:sharedrandomness}).

In our resource theory, the channels in $\free$ are the free objects.
In contrast, the quantum channels $\channel{N}$ not in $\free$ are genuine resources, as they require a quantum memory in their implementation.
Such channels are thus in the quantum domain, and by definition are not EB: they preserve the entanglement in at least {\em some} states $\rho_\mathrm{AR}$.

\begin{figure}
  \includegraphics{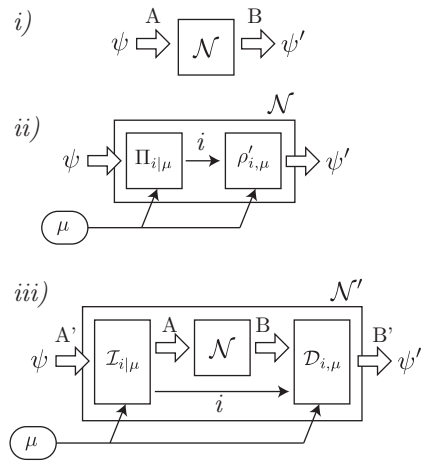}
  \caption{
    \label{fig:ebprocessors}
    i) The channel $\channel{N}:\density{A}\to\density{B}$ in our resource theory.
    ii) A free channel decomposed as a quantum-classical-quantum, or measure-and-prepare channel using POVM elements $\{\Pi_{i|\mu}\}$ and a family of states $\rho'_{i,\mu}$.
    iii) The free transformation of a channel $\channel{N}$ into $\channel{N}'$.
  }
\end{figure}

The last piece of our resource theory concerns the transformations of quantum channels.
The transformations $\Lambda$ we consider are the \textit{quantum supermaps}~\cite{Chiribella2008,Bisio2012} that transform a quantum channel $\channel{N}: \density{A} \to \density{B}$ into another quantum channel $\channel{N}': \density{A'} \to \density{B'}$, possibly related to different Hilbert spaces $\hilbert{A'}$ and $\hilbert{B'}$.
We allow the following ingredients in these supermaps:
\begin{itemize}
\item the use of preexisting randomness $\mu$,
\item a quantum operation $\density{A'} \to \density{A}$ before the use of the channel $\channel{N}$, possibly selected by $\mu$; in general, such operation can also produce a classical outcome labeled by $i$, that is we allow quantum instruments,
\item the storage of an arbitrary amount of classical information in the memory,
\item a quantum operation $\density{B} \to \density{B'}$, possibly selected using the classical information $i$ transmitted and the preexisting randomness $\mu$.
\end{itemize}
The general form of a free transformation $\Lambda$ thus has the following components (diagram in Fig.~\ref{fig:ebprocessors}iii):
\begin{equation}
  \label{eq:freetransformation}
  \channel{N}' = \Lambda[\channel{N}] = \sum_\mu \pi(\mu) \sum_i \mD_{i,\mu}\on{B'\from B} \circ \mN\on{B\from A} \circ \mI_{i|\mu}\on{A\from A'}\;.
\end{equation}
For each $\mu$, the quantum instrument $\{\mI_{i|\mu}\}_i$ is represented by a family of completely positive (CP) maps $\mI_{i|\mu}: \density{A} \to \subdensity{A'}$; while each $\mI_{i|\mu}$ is not necessarily trace-preserving, their sum $\sum_i \mI_{i|\mu}$ is trace preserving for each $\mu$.
The quantum operations $\{ \mD_{i,\mu} \}$ are indexed by $i$ and $\mu$, and for each pair $(i,\mu)$, the operator $\mD_{i,\mu}: \density{B'}\to\density{B}$ is a CPTP map.
Supermaps in the sense of~\eqref{eq:freetransformation} are {\em free} in the sense that they only require a classical memory; as a consequence, $\Lambda$ maps EB channels $\channel{N} \in \free$ to EB channels $\channel{N}' \in \free$.
To distinguish our free transformations from the set of all possible supermaps, we write $\freeop$ the set of all transformations of the form~\eqref{eq:freetransformation}, that we name {\em classically correlated supermaps}.

In the present case, we allow unbounded classical communication as part of the free resources.
Thus, as discussed in Appendix~\ref{app:sharedrandomness}, the establishment of the preexisting randomness $\mu$ can be embedded in the quantum instrument $\{\mathcal{I}_i\}$, a fact that we use later to simplify the notation in the proofs.

The set of free resources $\free$ is defined by the operational fact that the free channels break entanglement.
It is also reasonable that free resources constitute a convex set: the mere act of convexly mixing two ``useless'' channels should not give rise to a ``useful'' one.
This is indeed the case for EB channels.
However, the choice of the set $\freeop$ of free transformations that preserve $\free$ is not necessarily unique.
For example, one could consider, instead of classically correlated supermaps, only pre-processings or, alternatively, only post-processings: these would also leave the set of entanglement-breaking channels invariant.
However, such choices would not allow the preparation of any EB channel for free: in this sense, the choice we make here for $\freeop$ seems to be the most natural one in the present setting.

The fact that free operations are not (logically speaking) uniquely fixed by the set of free resources is a common feature shared by many resource theories.
For example, in the resource theory of entanglement, the set of free states is uniquely given by the separable states; however, this does not automatically single out a set of free transformations, which can be taken as the set of local quantum operations correlated by two- or one-way classical communication, or just by preexisting classical correlations, or even by classical correlations without causal order~\cite{Akibue2017} (the resulting transformations corresponding to separable operations).
In our case, we built $\freeop$ from the most general sequence of operations that respects the temporal order and grants any form of classical memory for free.
Smaller classes exist, such as supermaps correlated only by preexisting randomness, but those would correspond to artificial restrictions on the quantum memory users.
It is unclear whether \textit{larger} classes would exist after relaxing restrictions on the temporal order, however, it is doubtful that a resource theory of quantum memories would be meaningful in that context.

Straightforwardly, we observe that any EB channel $\channel{N} \in \free$ can be mapped to any other EB channel $\channel{N}' \in \free$ using a suitable $\Lambda \in \freeop$.
We also observe that memories that preserve exactly their input state are represented by the identity channel, and are the most powerful resources in our theory (of course, identity channels do not belong to $\free$).
Indeed, if we have at disposal a memory whose operation is the identity channel $\id: \density{A} \to \density{A}$, any memory $\channel{N}: \density{A} \to \density{B}$ can be simulated by using the free transformation $\Lambda[\id] = \mD \circ \id$: in this supermap, we do not use a quantum instrument, but simply apply the operation $\mD = \channel{N}$ after storage in $\id$ to simulate the processing of $\channel{N}$.
The transformations $\Lambda \in \freeop$ classify the relative power of quantum memories, as formalized by the following definition.
\begin{definition}
  \label{def:sufficient}
  We write $\channel{N} \sufficient \channel{N}'$ whenever there exists a classically correlated supermap $\Lambda \in \freeop$ such that $\channel{N}' = \Lambda[\channel{N}]$.
\end{definition}
The transformations $\freeop$ induce a partial order on quantum channels.
Given two channels $\channel{N}$ and $\channel{N}'$, it can be that neither $\channel{N} \sufficient \channel{N}'$ nor $\channel{N}' \sufficient \channel{N}$ hold; however, if $\channel{N} \sufficient \channel{N}'$ and $\channel{N}' \sufficient \channel{N}''$, then $\channel{N} \sufficient \channel{N}''$.

As it often happens with generalized resource theories, the partial order $\sufficient$ can be studied within the framework of statistical comparison in statistical decision theory, in the sense of Blackwell~\cite{Blackwell1953} and Le Cam~\cite{Cam1964}.
These ideas have been recently extended to the quantum setting~\cite{Buscemi2012a} and constitute the basis on which semiquantum games were originally introduced~\cite{Buscemi2012}.
More recently, the theory of (quantum) statistical comparison have been successfully applied to quantum information theory~\cite{Buscemi2014,Buscemi2016}, quantum open systems dynamics~\cite{Buscemi2014a,Buscemi2016a} and quantum thermodynamics~\cite{Buscemi2015,Buscemi2017,Buscemi2016b,Gour2017}.
One of the results of this work is to construct a framework to apply quantum statistical decision theory to the comparison of quantum channels and memories.

Note that the partial order $\sufficient$ distinguishes qualitatively between free and nonfree resources.
Quantitative statements about the usefulness of a resource are made using {\em monotones}, which are in this case real-valued functions of quantum channels, $\channel{N} \mapsto M(\channel{N})$, such that:
\begin{equation}
  \label{eq:monotone}
  \channel{N} \sufficient \channel{N}' \implies M(\channel{N}) \ge M(\channel{N}')\;.
\end{equation}

In particular, let $\channel{N}$ and $\channel{N}'$ be two free resources (i.e., two EB channels).
As $\channel{N} \sufficient \channel{N}'$ and $\channel{N}' \sufficient \channel{N}$, it follows that $M(\channel{N}) = M(\channel{N}')$.
Thus, any given monotone has a constant value on the set of free resources.

A {\em complete family of monotones} is a set $\{M_i\}_{i\in I}$ that completely characterizes the partial order $\sufficient$:
\begin{equation}
  \channel{N} \sufficient \channel{N}' \iff M_i(\channel{N}) \ge M_i(\channel{N}'), \quad \forall i \in I\;.
\end{equation}

As our resource theory allows convex mixtures of free and nonfree resources, the absolute and generalized robustness~\cite{Vidal1999,Steiner2003} can straightforwardly be computed~\footnote{For that purpose, the set of EB channels can be characterized using the symmetric extensions~\cite{Doherty2004} of their Choi, or equivalently, the existence of a channel that performs $1 \to N$ approximate symmetrical cloning~\cite{Bae2006,Chiribella2006a}, and whose restriction to a single output reproduces the original channel for all $N$. This hierarchy of tests can be formulated using semidefinite programming.}, as it has been done in other resource theories~\cite{Gallego2015,Napoli2016,SpekkensInPrep}.
The Schmidt number of a memory, as defined by the Schmidt number of the corresponding Choi operator, is a monotone with a straightforward operational interpretation: it corresponds to the maximal quantum dimension preserved by the channel~\cite{Chruscinski2006}.
Its practical computation, however, is outside the scope of the present manuscript and will be presented in a future work~\cite{ResourceTheoryInPrep}.

In the present work, we focus on a family of monotones that can be measured using minimal experimental assumptions.
As the family is complete, for any $\channel{N} \notin \free$, there exists a monotone that distinguishes $\channel{N}$ from free resources.
Those monotones are built on the semiquantum signaling games detailed in the next section.

\section{Semiquantum signaling games and existence of a test}
\label{sec:existence}
We come back to the semiquantum tests described in the introduction and put them on a formal footing.
\begin{definition}
  A \emph{semiquantum signaling scenario} is a tuple $\ssscenario = (\set{X},\set{Y},\set{B},\inputset{X},\inputset{Y})$ where:
  \begin{enumerate}
  \item $\set{X}=\{x\}$ and $\set{Y}=\{y\}$ are two (finite) index sets for the referee's questions;
  \item $\set{B}=\{b\}$ is a (finite) index set for Abby's answers;
  \item $\inputset{X}=\{\xi_x \in \density{X}:x\in\set{X} \}$ and $\inputset{Y}=\{\psi_y\in\density{Y}:y\in\set{Y} \}$ are two families of quantum states on $\hilbert{X}$ and $\hilbert{Y}$.
  \end{enumerate}
\end{definition}
Having fixed a semiquantum signaling scenario $\ssscenario$, a \emph{semiquantum signaling game} is played as follows:
\begin{enumerate}
\item at some time $t=t_0$, the referee randomly chooses an initial question $x\in\set{X}$ and sends the first state $\xi_x$ to Abby;
\item at a later time $t=t_1=t_0+\delta$ ($\delta>0$), the referee chooses another question $y\in\set{Y}$ and sends the second state $\psi_y$ to Abby;
\item Abby replies with an answer $b\in\set{B}$;
\item a payoff function $\payoff:\set{B}\times\set{X}\times\set{Y}\to\mathbb{R}$, publicly announced before the game started, decides the values $\payoff(b,x,y)$, namely, how much answer $b$ earns or costs Abby in the face of questions $x$ and $y$.
\end{enumerate}
A semiquantum signaling game is completely described by the payoff function $\payoff$ in the (implicit) context of a scenario $\ssscenario$.
By comparing the relative frequencies of the questions posed and the answers given, the referee can estimate the correlation $p(b|x,y)$.
Note that we do not address the effects of finite statistics in the present work, and use for now ideal distributions for $p(b|x,y)$.
We prescribe the computation of Abby's average payoff to be:
\begin{align}
  \label{eq:expect-payoff}
  \sum_{x,y}\payoff(b,x,y)p(b|x,y)\;,
\end{align}
which corresponds, up to a multiplicative factor, to a uniformly random distribution of inputs $p(x,y)$.

The resource that Abby can utilize in order to maximize her expected payoff is a given channel $\channel{N}:\density{A}\to\density{B}$, taking quantum states defined on an input Hilbert space $\hilbert{A}$ at time $t_0$, to quantum states defined on an output Hilbert space $\hilbert{B}$ at time $t_1$.
We assume that Abby uses the same channel $\channel{N}$ only once in each round of the game.
Besides the channel $\channel{N}$, Abby can freely use any amount of classical memory she wants.
Formally, this amounts to letting Abby transform the channel $\channel{N}$ into $\channel{N}': \density{X}\to\density{B'}$ by any of the (free) transformations $\Lambda \in \freeop$.
While doing so, Abby also adapts the dimension of the input in case the spaces $\hilbert{A}$ and $\hilbert{X}$ are not isomorphic.
When also the second quantum question is received, Alice may jointly measure this second question (state) along with the output of the transformed channel $\channel{N}'$.

Referring to Appendix~\ref{app:sharedrandomness}, we embed randomness in the classical communication and write any admissible strategy for Abby for the resource $\channel{N}$ as:
\begin{align}
  \label{eq:admissiblefull}
p_\channel{N}(b|x,y)=\sum_i \Tr{\left\{(\mD\on{B}_i \circ \channel{N}\on{A} \circ\mI_i\on{X})(\xi_x\on{X})\otimes\psi_y\on{Y}\right\}\ {B'}_{b|i}\on{B'Y}}\;,
\end{align}
where $\{{B'}_{b|i}\on{B'Y}\}$ represents a POVM for each $i$, in the sense that $\sum_b {B'}_{b|i}=\openone$ for all $i$.
Since the action of the channels $\mD_i^B$ can be absorbed in the POVMs $\{{B'}_{b|i}\on{B'Y}\}$ we will, from now on, avoid writing them explicitly.

Given a resource channel $\channel{N}$, we obtain our characterization of admissible strategies.
\begin{definition}\label{Dfn:Admissible}
  The correlations $p_\channel{N}(b|x,y)$ are admissible for the resource $\channel{N}$ in the scenario $\ssscenario$ if there exists a quantum instrument $\{ \mI_i\on{X} : \density{X} \to \subdensity{A} \}$, and a family of measurements $\{ B_{b|i}\on{BY} \}$ such that:
\begin{align}
  \label{eq:admissible}
  p_\channel{N}(b|x,y)=\sum_i \Tr{\left\{(\channel{N}\on{A} \circ\mI_i\on{X})(\xi_x\on{X})\otimes\psi_y\on{Y}\right\}\ B_{b|i}\on{BY}}.
\end{align}
\end{definition}
For a given channel $\channel{N}$ and a given semiquantum signaling scenario $\ssscenario=(\set{X},\set{Y},\set{B},\inputset{X},\inputset{Y})$, we define the set
\begin{equation}
\label{eq:admissibleset}
\set{S}(\channel{N},\ssscenario)\defeq\{p_\channel{N}(b|x,y):p_\channel{N}(b|x,y)\text{ is admissible} \}\;,
\end{equation}
where by ``admissible'' we mean that $p_\channel{N}(b|x,y)$ can be written as in Eq.~(\ref{eq:admissible}) for varying instruments and POVMs.
The set $\set{S}(\channel{N},\ssscenario)$ is a convex, closed and bounded subset of $\mathbb{R}^N$ for $N={|\set{B}|\cdot|\set{X}|\cdot|\set{Y}|}$ (see Appendix~\ref{app:sharedrandomness}).

It is possible to compute the utility of the resource channel $\mN$ in any semiquantum signaling game by computing the maximal expected payoff over all admissible strategies:
\begin{equation}
  \label{eq:optimalpayoff}
  \payoff^*(\channel{N})\defeq\max\sum_{x,y}\payoff(b,x,y)p_\channel{N}(b|x,y)\;,
\end{equation}
where the optimization is made over all $p_\channel{N}(b|x,y)\in\set{S}(\channel{N},\ssscenario)$.

\begin{proposition}
  Let $\ssscenario$ and $\payoff$ define a semiquantum signaling game.
  Then $\payoff^*: \channel{N} \mapsto \payoff^*(\channel{N})$, as defined in Eq.~\eqref{eq:optimalpayoff} is a monotone in the sense of Eq.~\eqref{eq:monotone}.
\end{proposition}
\begin{proof}
  Let $\channel{N}$ and $\channel{N}'$ be two channels such that $\channel{N} \sufficient \channel{N}'$.
  Then, any strategy that achieves $\payoff^*(\channel{N}')$ can also be realized using $\channel{N}$, as there is a transformation $\Lambda \in \freeop$ that Abby can use as part of her strategy~\eqref{eq:admissiblefull} to transform $\channel{N}$ into $\channel{N}'$.
  Thus $\payoff^*(\channel{N}) \ge \payoff^*(\channel{N}')$ as required.
\end{proof}

As $\payoff^*$ is a monotone, it has a constant value $\payoff_\text{EB}$ on all entanglement-breaking channels, which we define as the {\em entanglement-breaking threshold} for the signaling semiquantum game defined by $\ssscenario$ and $\payoff$.
For a given $\channel{N}$, any admissible strategy $p_\channel{N}(b|x,y)$ that achieves a payoff greater than $\payoff_\text{EB}$ certifies that $\channel{N}$ is not entanglement-breaking.
This holds even if the admissible strategy does not achieve the maximal expected payoff $\payoff^*(\channel{N})$, which makes the procedure robust when dealing with imperfect experimental implementations.

Given a non-EB $\channel{N}$, the main question is to find a signaling semiquantum game able to certify that $\channel{N} \notin \free$.
We answer that question by proving a stronger result: signaling semiquantum games form a complete family of monotones for the resource theory of quantum memories.

\begin{theorem}
  \label{th:completefamily}
  Let $\channel{N}:\density{A}\to\density{B}$ and $\channel{N}':\density{A'}\to\density{B'}$ be two channels.
  If $\payoff^*(\channel{N}) \ge \payoff^*(\channel{N}')$ for all signaling semiquantum games, then $\channel{N} \sufficient \channel{N}'$.
\end{theorem}
\begin{proof}
  In Appendix~\ref{app:proofthm1}.  
\end{proof}

The answer to our question comes from the following corollary.

\begin{corollary}
  \label{Cor:Witness}
  Let $\channel{N}:\density{A}\to\density{B}$.
  Then $\channel{N}$ is in the quantum domain $(\channel{N} \notin \free)$ if and only if there exists a semiquantum signaling game such that $\payoff^*(\channel{N})>0$ while $\payoff_{\EB}=0$.
\end{corollary}

\begin{proof}
  Let $\channel{N}'\in\free$ be any entanglement-breaking channel.
  As $\channel{N}' \nrightarrow \channel{N}$, there exists a semiquantum signaling game $(\ssscenario, \payoff)$ such that $\payoff^*(\channel{N}) > \payoff_\text{EB}$ by the converse of Theorem~\ref{th:completefamily}.
  Finally, $\payoff_{\EB}$ can be made zero by simply shifting the original payoff function by a fixed constant.
\end{proof}

Thus, any channel in the quantum domain can be verified by some semiquantum signaling game.
In the next section, we provide a constructive proof of this fact.

\section{Construction of experimentally friendly semiquantum signaling games}
\label{sec:construction}

We now turn to explicit constructions of tests for channels in the quantum domain, and provide a constructive proof of Corollary~\ref{Cor:Witness} independent of Theorem~\ref{th:completefamily}.
We then illustrate this construction by an example, and finally show that our construction is resistant to losses.

First, let us consider the particular semiquantum correlations that encode the essential knowledge about a channel.

\begin{definition}
  \label{def:signature}
  Given a channel $\channel{N}:\density{A}\to\density{B}$, we define its {\em signature scenario}
  \[
    \sigscenario{\channel{N}}=(\set{X},\set{Y},\set{B},\inputset{X},\inputset{Y})
  \]
  where $\inputset{X}=\{\xi_x\on{X}:x\in\set{X} \}$ and $\inputset{Y}=\{\psi_y\on{Y}:y\in\set{Y} \}$ are tomographically complete, respectively, for the Hilbert spaces of the channel: $\hilbert{X}\cong\hilbert{A}$ and $\hilbert{Y}\cong\hilbert{B}$.
  We set $|\set{B}|=(\dim\hilbert{B})^2$.
  Note that the completeness of $\inputset{X},\inputset{Y}$ implies that $|\set{X}|=(\dim\hilbert{A})^2$ and, $|\set{Y}|=(\dim\hilbert{B})^2$.

  The {\em signature correlation} $\sigcorrelation{\channel{N}}(b|x,y)$ is given by:
  \begin{align}
    \label{eq:signaturecorrelation}
    \sigcorrelation{\channel{N}}(b|x,y)=\Tr{\left\{{\channel{N}}\on{A}(\xi_x\on{X})\otimes\psi_y\on{Y}\right\}\ {B}_b\on{BY}}\;,
  \end{align}
  where we chose the POVM $\{B_b\on{BY}\}$ to be a complete Bell measurement (remember that $\hilbert{A}\cong \hilbert{X}$ and $\hilbert{B}\cong \hilbert{Y}$).
  For later use, we assume that the first POVM element $B_1 = \Phi_+\on{BY}$ is the maximally entangled state in the computational basis:
  \begin{equation}
    \label{eq:maximallyentangled}
    \Phi_+ = \frac{1}{d}\sum_{ij} \ketbra{ii}{jj}\;,
  \end{equation}
  where $d$ is the corresponding Hilbert space dimension.
  Under these conditions, $\sigcorrelation{\channel{N}}$ contains full tomographic data about the channel $\channel{N}$.
\end{definition}

The signature correlation of $\channel{N}$ describes fully its behavior, as it essentially amounts to a sort of quantum process tomography.
This fact is used in the proof of Theorem~\ref{th:completefamily} in Appendix~\ref{app:proofthm1}.
We now move to the explicit construction of semiquantum signaling games that witness channels in the quantum domain.
To characterize such channels, we make use of the Choi-Jamiolkowski~\cite{Choi1975,Jamiolkowski1972} representation $\choi{\channel{N}} \in \densitytwo{A}{B}$ of $\channel{N}: \density{A}\to\density{B}$:
\begin{equation}
  \label{eq:choi}
  \choi{\channel{N}} \defeq (\openone \otimes \channel{N})(\Phi_+)\;,
\end{equation}
where $\Phi_+$ is the maximally entangled state defined in~\eqref{eq:maximallyentangled}, such that:
\begin{equation}
  \label{eq:choiback}
  \Tr{\channel{N}(A)\ B} = d \Tr{\choi{\channel{N}}\ (A^\top \otimes B)}
\end{equation}
where $d =\dim\hilbert{A}$.
The superscript $\top$ denotes the transposition with respect to the computational basis.

\begin{proposition}
  \label{prop:construction}
  Let $\channel{N}: \density{A}\to\density{B}$ be a non-EB channel.
  Then its Choi-Jamiolkowski state $\choi{\channel{N}} = (\openone \otimes \channel{N})(\Phi_+)$ is necessarily entangled~\cite{Jiang2013}.
  Let $W$ be an entanglement witness such that $\Tr{W \choi{\channel{N}}} > 0$ while $\Tr{W \rho} \le 0$ for any separable state $\rho \in \densitytwo{A}{B}$.
  Let $\sigscenario{\channel{N}}$ be the signature scenario associated with $\channel{N}$, in which $W$ has the decomposition:
  \begin{equation}
    \label{eq:witnessdecomposition}
    W = \sum_{xy} \omega_{xy} (\xi_x^\top \otimes \psi_y^\top)\;.
  \end{equation}
  Then the payoff
  \begin{equation}
    \payoff(1,x,y) = \omega_{xy}, \quad \payoff(b>1,x,y) = 0
  \end{equation}
  defines a semiquantum signaling game that satisfies the conditions of the Corollary: $\payoff^*(\channel{N}) > 0$ while $\payoff_{\EB} = 0$.
\end{proposition}

\begin{proof}
  In Appendix~\ref{app:proofconstruction}.
\end{proof}

With this construction, the payoff makes only use of the coefficients $\sigcorrelation{\channel{N}}(b=1|x,y)$, and thus the complete Bell measurement can be replaced by a partial Bell measurement to reduce the experimental requirements.
In Appendix~\ref{app:sparsewitness}, we show how to reduce the number of input pairs $(x,y)$ for which statistics need to be collected to $d^2 + 3$, where $d=\min(\dim\hilbert{X}, \dim\hilbert{Y})$, which has better scaling than the use of tomographically complete sets of inputs (at least $d^4$ input pairs).

\subsection{Example: qubit depolarizing channel}
As an example, consider the qubit-qubit depolarizing channel
\begin{equation}
  \channel{N}_\nu(\rho\on{A}) = \nu \rho + (1-\nu) \frac{\openone}{2}\;.
\end{equation}
As \channel{N} acts on a qubit space, we use the tomographically complete set of quantum inputs
\begin{equation}
  \xi_x = U_x \tau U_x^\dagger, \quad \psi_y = U_y \tau U_y^\dagger,
\end{equation}
where $U_1$, $U_2$, $U_3$ and $U_4$ are respectively the identity $\openone$ and the three Pauli matrices $\sigma_x$, $\sigma_y$, $\sigma_z$ while $\tau = \openone/2 + (\sigma_x + \sigma_y + \sigma_z)/\sqrt{12}\;$; note that $\set{X} = \set{Y}$.
According to our construction, we use a partial Bell measurement with $B_1 = \Phi_+$ and $B_2 = \openone - B_1$.
We get:
\begin{equation}
  \sigcorrelation{\channel{N}}(1|x,y) =
  \begin{cases}
    (1 - \nu)/4 & \mbox{if } x - y = 2 \mbox{ mod } 4, \\
    (3 + \nu)/12 & \mbox{otherwise}.
  \end{cases}
\end{equation}

In our example, we get:
\begin{equation}
  \choi{\channel{N}} = \nu \Phi_+ + (1-\nu) \openone/4\;.
\end{equation}

The entanglement witness related to $\choi{\channel{N}}$ is $W = \Phi_+ - \openone/2$.
According to its decomposition on $\xi_x$ and $\psi_y$, we obtain the payoff:
\begin{equation}
  \payoff(1,x,y) =
  \begin{cases}
    -5/8 & \mbox{if } x-y = 2 \mbox{ mod } 4, \\
    1/8 & \mbox{otherwise},
  \end{cases}
\end{equation}
such that the expected payoff value is $(3\nu - 1)/4$, which detects faithfully a channel in the quantum domain for $\nu > 1/3\approx 33\%$.
In contrast, tests based on the CHSH inequality can only detect channels in the quantum domain for $\nu > 2^{-1/2}\approx 71\%$, or $\nu > 2^{-1/4} \approx 84\%$ if two copies of the channel are used, according to the setup of Figure~\ref{fig:bell-test}.

\subsection{Robustness against losses}
\label{sec:loss}

Our construction is robust against isotropic losses and detection inefficiencies.
From the perspective of the correlation $p(b|x,y)$, all losses and detection inefficiencies that do not depend on the indices $x$ and $y$ can be modeled as an erasure channel:
\begin{equation}
  \channel{E}_\eta(\rho) = \eta \rho + (1 - \eta) \vacuum,
\end{equation}
applied after the use of the channel $\channel{N}$, such that we test effectively
\begin{equation}
  \channel{N}' = \channel{E}_\eta \circ \channel{N}\;.
\end{equation}

In the spirit of erasure channels, we assume that $\vacuum$ lies outside the range of $\channel{N}$, so that losses are always identified.

\begin{proposition}
  Let $\channel{E}_\eta$ an erasure channel for $\eta > 0$.
  Then $\channel{E}_\eta \circ \channel{N}$ is in the quantum domain if and only if $\channel{N}$ is in the quantum domain.
\end{proposition}

\begin{proof}
  Let $\choi{\channel{N}}$ be the Choi-Jamiolkowski representation~\eqref{eq:choi} of $\channel{N}$.
  The representation of $\channel{E}_\eta \circ \channel{N}$ is:
  \begin{equation}
    \label{eq:choierased}
    \choi{\channel{E}_\eta\circ\channel{N}} = \eta \choi{\channel{N}} + (1-\eta) \openone \otimes \vacuum.
  \end{equation}
  When $\channel{N}$ is EB, the state~\eqref{eq:choierased} is a mixture of separable states, and thus $\channel{E}_\eta\circ\channel{N}$ is EB.
  Assume now that $\channel{N}$ is non-EB.
  Then there exists an entanglement witness $W$, without support on $\vacuum$, such that $\Tr{W \choi{\channel{N}}} > 0$ while $\Tr{W \rho_\text{SEP}} \le 0$.
  The result follows by applying $W$ on the Choi of $\channel{E}_\eta\circ\channel{N}$: $\Tr{W \choi{\channel{E}_\eta\circ\channel{N}}} = \eta \Tr{W \choi{\channel{N}}} > 0$.
\end{proof}

The constructive approach of Proposition~\ref{prop:construction} can then be straightforwardly applied.
Assume that $\payoff(b,x,y)$ corresponds to a semiquantum signaling game for $\channel{N}$.
When testing the channel $\channel{N}' = \channel{E}_\eta \circ \channel{N}$ with isotropic erasure, we add an element $B_0\on{BY} = \vacuum \otimes \openone$ to the measurement $\{B_b\}$ such that the erased state $\vacuum$ is sent to a new measurement outcome $b=0$.
Let $p(b|x,y)$ be the correlation of the channel $\channel{N}$ that obtained a payoff $\left<\payoff\right> > 0$.
With the above scheme, the correlation of $\channel{N}'$ is readily obtained:
\begin{equation}
  p'(1|x,y) = \eta p(1|x,y)\;.
\end{equation}
Setting accordingly $\payoff(0,x,y) \defeq 0$ for the new measurement outcome, the average payoff on $\channel{N}'$ is:
\begin{equation}
  \sum_{xy} \payoff(1,x,y) p'(1|x,y) = \eta \left< \payoff \right>
\end{equation}
and thus demonstrates the non-EB nature of $\channel{N}'$ for all $\eta>0$.

\section{Conclusion}

In our paper, we define a class of tests that can faithfully verify the quantum nature of memories with minimal assumptions.
To do so, we provide a qualitative resource theory necessary to distinguish (non)-entanglement-breaking quantum channels as different classes of resources, mirroring the resource theory of entangled states and their transformations under LOCC.
Allowing the classical storage of any amount of information for free, we identify the nontrivial resources as the channels preserving entanglement.
By the use of classically correlated pre/post processing supermaps, we define how quantum channels can be transformed, and accordingly show the existence of a partial order on channels.
We single out the class of entanglement-breaking channels that operate by the storage of classical information.
We complete our resource theory by defining the monotones relevant to the quantitative study of quantum memories.
Second, we translate the idea of Buscemi~\cite{Buscemi2012} to the temporal setting and construct semiquantum games for temporal correlations; by showing that the maximal expected payoffs of such games form a complete family of monotones, we demonstrate the existence of measurement-device-independent (MDI) tests for memories in the quantum domain.
We finally provide a construction of such tests, and show that the resulting MDI witnesses can certify all memories in the quantum domain, even when facing arbitrary losses.

Our work opens new research avenues regarding the classification of quantum channels.
The monotones detailed in this work are motivated by experimental tests that can be constructed with minimal assumptions.
As recently explored in the spatial context~\cite{Verbanis2016,Supic2017a,Shahandeh2017,Shahandeh2017}, the semiquantum framework and the related MDI witnesses can provide bounds on a variety of operational entanglement measures.
This relation should also hold in the temporal context: we leave it as an open question to find the relation between our family of monotones and the various quantitative measures already defined on quantum channels (like, e.g., channel capacities or other entropic quantities). 

We also note that the semiquantum framework applies equally well to spatial and temporal correlations, without encountering the issues that plague the description of time-like joint states~\cite{Horsman2017}.
By treating space and time on an equal footing, this framework requires only minimal assumptions and is well-suited to the examination of quantum causal structures when facing arbitrary causal orders~\cite{Oreshkov2012}.

As a final comment, our results cater for memories acting on finite dimensional Hilbert spaces. This restriction eases considerably the derivation: for example, our construction relies on Bell measurements, whose generalization to e.g. continuous degrees of freedom has a quite different nature~\cite{Hofer2013}.
As some experimental tests of quantum memories involve continuous degrees of freedom~\cite{Killoran2012,Yang2014a}, we leave as an open question the generalization of our results to infinite dimensional systems.

\section*{Acknowledgments}

We acknowledge discussions with Cyril Branciard, Nicolas Gisin and Miguel Navascu{\'e}s.
Research at Perimeter Institute is supported by the Government of Canada through Industry Canada and by the Province of Ontario through the Ministry of Research and Innovation.
This publication was made possible through the support of a grant from the John Templeton Foundation. 
D.R. was supported by the SNSF Early Postdoc. Mobility fellowship P2GEP2\_162060, F.B. was supported in part by JSPS KAKENHI, grant no. 26247016 and no. 17K17796.
Y.C.L. is supported by the Ministry of Science and Technology of Taiwan (Grant No.104-2112-M-006-021-MY3).
Part of this work was initiated during the Nagoya Winter Workshop NWW2013.

\bibliography{memdiew}

\appendix

\section{Use of randomness in the resource theory of channels}
\label{app:sharedrandomness}
Preexisting randomness, represented in our formulas by the index $\mu$, plays an important role in the definition of entanglement-breaking channels~\eqref{eq:entanglementbreaking} and free transformations~\eqref{eq:freetransformation}, as it is the basic ingredient that provides convexity to sets of objects.
In our case, however, the use of randomness can be embedded in the quantum operations and classical communication, and thus the notation used to describe related objects can be simplified.

\begin{proposition}
  Let $\Lambda$ be a free transformation described by:
  \begin{equation}
    \channel{N}' = \Lambda[\channel{N}] = \sum_\mu \pi(\mu) \sum_i \mD_{i,\mu}\on{B'\from B} \circ \mN\on{B\from A} \circ \mI_{i|\mu}\on{A\from A'}\;,
  \end{equation}
  where $\mu$, $\{\mI_{i|\mu}\}$, $\{\mD_{i,\mu}\}$ are as described in the main text.
  Then there is an equivalent free transformation $\Lambda'$ that does not require randomness:
  \begin{equation}
    \channel{N}' = \Lambda'[\channel{N}] = \sum_{i'} {\mD'}_{i'}\on{B'\from B} \circ \mN\on{B\from A} \circ {\mI'}_{i'}\on{A \from A'}.
  \end{equation}
\end{proposition}

\begin{proof}
  First, we remark that our definitions implicitly assume that $\mu$ is discrete and that $\pi(\mu)$ is a probability distribution.
  This does not hinder full generality as the input and output spaces $\hilbert{A}$ and $\hilbert{B}$ are finite; then the space of free transformations between $\hilbert{A}$ and $\hilbert{B}$ has finite dimension, and thus $\mu$ can always be taken finite (by Caratheodory's theorem; see a related discussion in~\cite{Rosset2017a}).
  To construct our equivalent transformation $\Lambda'$, we write $i' = (i, \mu)$, and define:
  \begin{equation}
    \mI'_{i'}(\rho) = \mI'_{i,\mu}(\rho) \defeq \pi(\mu) \mI_{i|\mu}(\rho).
  \end{equation}
  By normalization of $\pi(\mu)$, we have $\Tr{\sum_{i,\mu} \mI'_{i,\mu}(\rho)} = 1$ for all $\rho\in\density{A}$; and the elements are still CP maps as $\pi(\mu) \ge 0$.
  Thus $\mI'_{i'}$ is a proper quantum instrument.
  We obtain the desired result by defining $\mD'_{i'} = \mD_{i,\mu}$.
\end{proof}

Two important consequences follow.
First, as all EB channels can be obtained by applying appropriate free transformations on another EB channel, we can remove the use of preexisting randomness from~\eqref{eq:entanglementbreaking} as well.
Second, the set of admissible correlations for $\channel{N}$, $\set{S}(\channel{N},\ssscenario)$ is convex, as a consequence of the convexity of admissible strategies in Eq.~(\ref{eq:admissible}).

\section{Proof of Theorem~\ref{th:completefamily}}
\label{app:proofthm1}

We now demonstrate that $\payoff^*(\channel{N}) \ge \payoff^*(\channel{N}')$ for all signaling semiquantum games implies that $\channel{N} \sufficient \channel{N}'$.
Our proof proceeds as follows.

In a fixed signaling semiquantum scenario $\ssscenario$, we first show that if $\channel{N}$ achieves a maximal expected payoff at least as good as $\channel{N'}$ for any $\payoff$,  then $\channel{N}$ can reproduce the full range of correlations admissible for $\channel{N}'$.
We then show that \textit{any} channel that can reproduce the full range of semiquantum correlations of $\channel{N}'$ can be, in fact, transformed into $\channel{N}'$ via some free transformation.

\begin{proposition}
  Let $\ssscenario$ be a fixed signaling semiquantum scenario.
  If $\payoff^*(\channel{N}) \ge \payoff^*(\channel{N}')$ for all $\payoff$, then $\channel{N}$ can reproduce all correlations of $\channel{N}'$:
  \begin{equation}
    \set{S}(\channel{N},\ssscenario) \supseteq \set{S}(\channel{N}',\ssscenario).
  \end{equation}
\end{proposition}
\begin{proof}
Given a payoff function $\payoff$, the expected payoff of $\channel{N}$ can be rewritten as
\begin{equation}
  \payoff^*(\channel{N})= \max_{\vec{p}\;\in \set{S}(\mN,\ssscenario)}\vec{p}\cdot\vec{\wp}\;,
\end{equation}
where we use the notation $\vec{\payoff}$ to denote the real vector $(\payoff_{x,y,a}=\wp(x,y,a))\in\mathbb{R}^N\;$.
Correspondingly, the condition $\payoff^*(\channel{N})\ge\payoff^*(\channel{N}')$, can be rewritten as follows:
\begin{equation}
  \max_{\vec{p}\;\in \set{S}(\mN,\ssscenario)}\vec{p}\cdot\vec{\wp}\ge \max_{\vec{p}\;'\in \set{S}(\mN',\ssscenario)}\vec{p}\;'\cdot\vec{\wp}\;,
\end{equation}
for all $\payoff$.
The payoff vector $\vec{\payoff}$ can be any vector in $\mathbb{R}^N$.
Since $\set{S}(\channel{N},\ssscenario)$ and $\set{S}(\channel{N}',\ssscenario)$ are convex (see Appendix~\ref{app:sharedrandomness}), the separation theorem for convex sets~\cite{Rockafellar1970} implies that the above inequality is equivalent to
\begin{equation}
  \label{eq:convex-incl}
  \set{S}(\channel{N},\ssscenario)\supseteq \set{S}(\channel{N}',\ssscenario)\;,
\end{equation}
which proves the Proposition.
\end{proof}

We now make use of the signature correlation (Definition~\ref{def:signature}) for our target channel $\channel{N}'$.
It has the property that any channel that can reproduce it can be transformed into $\channel{N}'$.
\begin{proposition}
  Let $\sigcorrelation{\channel{N}'}$ be the signature correlation obtained in the scenario $\sigscenario{\channel{N}'}$ for the channel $\channel{N}'$:
\begin{align}
  \sigcorrelation{\channel{N}'}(b|x,y)=\Tr{\left\{{\channel{N}'}\on{A'}(\xi_x\on{X})\otimes\psi_y\on{Y}\right\}\ {B'}_b\on{B'Y}}\;,
\end{align}
Then any channel $\channel{N}$ whose admissible set contains $\sigcorrelation{\channel{N}'}$ can, in fact, be transformed into $\channel{N}'$: $\channel{N} \sufficient \channel{N}'$.
\end{proposition}

\begin{figure}
  \includegraphics{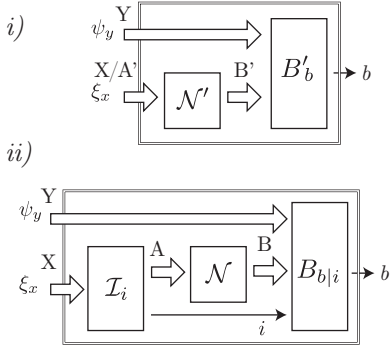}
  \caption{
    A diagram relating the uses of the channel $\channel{N}$ and $\channel{N}'$ to produce the signature correlation of $\channel{N}'$ in Eq.~\eqref{eq:prob-iden}.
    The correlations of both boxes enclosed in double lines are identical, and thus their behavior is essentially indistinguishable.
  }
  \label{fig:twochannels}
\end{figure}

\begin{proof}
  Assume that the channel $\mathcal{N}$ can reproduce the correlation $\sigcorrelation{\channel{N}'}(b|x,y)$.
  From Eq.~\eqref{eq:convex-incl}, we have an admissible strategy of $\sigcorrelation{\channel{N}'}$ using $\channel{N}$, comprising a quantum instrument $\{ \mathcal{I}_i\on{X\to A}\}$ and a family of measurements $\{ B_{b|i}\on{BY} \}$, for which
  \begin{align}
    \label{eq:prob-iden}
    p(b|x,y)&=\sum_i \Tr{\left\{(\channel{N}\on{A}\circ\mI_i\on{X})(\xi_x\on{X})\otimes\psi_y\on{Y}\right\}\ B_{b|i}\on{BY}}\nonumber\\
            &=\Tr{\left\{{\channel{N}'}\on{A'}(\xi_x\on{X})\otimes\psi_y\on{Y}\right\}\ {B'}_b\on{B'Y}}\;,
  \end{align}
  for all $x$, $y$ and $b$.
  The graphical translation of Eq.~\eqref{eq:prob-iden} is given in Figure~\ref{fig:twochannels}.
  As the sets of inputs are tomographically complete, the correlation $p(b|x,y)$ contains the data related to quantum tomography of the process.
  Thus, the quantum systems enclosed by the double line have the same external behavior, and are thus indistinguishable for our purposes.
  
The above arguments can be made rigorous as follows. Introducing another system $\hilbert{Y'}\cong \hilbert{Y}$, we now write $\psi\on{Y}_y=\PTr{Y'}{\Phi_+\on{YY'}\ (\openone\on{Y}\otimes\zeta_y\on{Y'})}$, where $\Phi_+\on{YY'}$ is the maximally entangled state between $\hilbert{Y}$ and $\hilbert{Y'}$.
Due to the completeness of the $\psi\on{Y}_y$ we have the completeness of the $\zeta_y\on{Y'}$.
Hence, Eq.~(\ref{eq:prob-iden}) can be written as an operator identity, i.e.,
\begin{align}
  \label{eq:almost-tele}
  &\sum_i \PTr{BY}{\left\{(\channel{N}\on{A}\circ\mI_i\on{X})(\xi\on{X}_x)\otimes\Phi_+\on{YY'}\right\}\ \left\{B_{b|i}\on{BY}\otimes\openone\on{Y'}\right\}}\nonumber\\
  &=\PTr{B'Y}{\left\{{\channel{N}'}\on{A'}(\xi_x\on{X})\otimes\Phi_+\on{YY'}\right\}\ \left\{{B'}_b\on{B'Y}\otimes\openone\on{Y'}\right\}}\;.
\end{align}

\begin{figure}
  \includegraphics{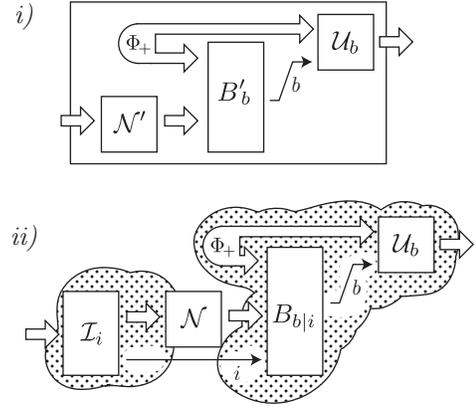}
  \caption{
    i)
    To correct the effect of the Bell measurement and the preparation of a maximally entangled state, it is sufficient to apply a unitary correction selected using the outcome of the Bell measurement.
    The resulting channel (bigger box) is equivalent to $\channel{N}'$.
    ii)
    As seen in Figure~\ref{fig:twochannels}, the same correction can be applied when using $\channel{N}$ to recover $\channel{N}'$.
    In dotted blob, the parts of the protocol that are regrouped to form a free transformation $\Lambda$ as defined in Section~\ref{sec:resourcetheory}.
  }
  \label{fig:extracting}
\end{figure}

Invoking now the protocol of quantum teleportation (Fig.~\ref{fig:extracting}i), we know that there exist unitary channels $\mU_b\on{Y'}$ such that
\begin{multline}
  \label{eq:channel-tele}
  {\channel{N}'}\on{A'}(\xi_x\on{X}) =
  \sum_b \mU\on{Y'}_b
  \Big(
  \PTrOp{B'Y}\Big[
    \Big\{
    {\channel{N}'}\on{A'}(\xi_x\on{X})\otimes\Phi_+\on{YY'}
    \Big\}\\
    \times \Big\{
    {B'}_b\on{B'Y}\otimes\openone\on{Y'}
    \Big\}
  \Big]
  \Big)\;,
\end{multline}
for all $x$, namely, using Eq.~(\ref{eq:almost-tele}),
\begin{multline}
  {\channel{N}'}\on{A'}(\xi_x\on{X})=
  \sum_{ib}\mU\on{Y'}_b
  \Big(
  \PTrOp{BY}\Big[
  \Big\{(\mN\on{A}\circ\mI_i\on{X})(\xi_x\on{X})\otimes\Phi_+\on{YY'}\Big\} \\
  \times \Big\{B_{b|i}\on{BY}\otimes\openone\on{Y'}\Big\}
  \Big]
  \Big)\;.
\end{multline}
Due to the completeness of $\xi_x\on{X}$, the above equation indeed holds as a \textit{map} (remember that $\hilbert{A'}\cong \hilbert{X}$):
\begin{multline}
  \label{eq:simulation}
  {\channel{N}'}\on{A'}(\bullet\on{A'})=
  \sum_{ib}\mU\on{Y'}_b
  \Big(\PTrOp{BY}\Big[
  \Big\{(\mN\on{A}\circ\mI_i\on{X})(\bullet\on{X})\otimes\Phi_+\on{YY'}\Big\}\\
  \Big\{B_{b|i}\on{BY}\otimes\openone\on{Y'}\Big\}
  \Big]\Big)\;.
\end{multline}
A schematic diagram of the above equation is given in Fig.~\ref{fig:extracting}.
Finally, introducing the maps $\mathcal{D}_i:B\to Y' \cong B'$ defined as
\begin{equation*}
  \mathcal{D}_i\on{B}(\bullet)\defeq\sum_{b}\mU\on{Y'}_b\left(\PTr{BY}{\left\{\bullet\on{B}\otimes\Phi_+\on{YY'}\right\}\ \left\{B_{b|i}\on{BY}\otimes\openone\on{Y'}\right\}} \right)\;,
\end{equation*}
we have that
\begin{equation*}
  {\channel{N}'}\on{B'\from A'}=\sum_i\mathcal{D}_i\on{B'\from B}\circ\mN\on{B\from A}\circ\mI_i\on{A\from A'},
\end{equation*}
which proves that $\mN \sufficient \mN'$.
\end{proof}

The proof of Theorem~\ref{th:completefamily} follows, as we proved that $\channel{N}$ can reproduce all correlations of $\channel{N}'$ as a consequence of $\payoff^*(\channel{N}) \ge \payoff^*(\channel{N}')$ for all games.

\section{Proof of Proposition~\ref{prop:construction}}
\label{app:proofconstruction}
  We first remark that, by construction of the signature correlation~\eqref{eq:signaturecorrelation}, we have, using~\eqref{eq:choiback}, that
  \begin{align}
    \sigcorrelation{\channel{N}}(1|x,y) & = \Tr{\left(\channel{N}(\xi_x\on{X}) \otimes \psi_y\on{Y}\right)\Phi_+\on{BY}} \nonumber \\
    & = d \Tr{\left(\choi{\channel{N}} \otimes \psi_y\on{Y}\right) \left(\xi_x^\top \otimes \Phi_+ \right)} \nonumber \\
    & = \Tr{\choi{\channel{N}}(\xi_x^\top \otimes \psi_y^\top)}.
  \end{align}
  Then, by construction,
  \begin{align}
    \sum_{xy} \payoff(1|x,y) \sigcorrelation{\channel{N}}(1|x,y) & = \sum_{xy} \omega_{xy} \Tr{\choi{\channel{N}}(\xi_x^\top \otimes \psi_y^\top)} \nonumber \\
                                                                 & = \Tr{\choi{\channel{N}} W} > 0\;.
  \end{align}
  Thus the optimal payoff when Abby has access to $\channel{N}$ is positive.
  Now, consider an EB channel.
  Without loss of generality, its correlations can be written, again using~\eqref{eq:entanglementbreaking}, as follows:
  \begin{equation}
    p_{\EB}(b|x,y) = \sum_i \Tr{\xi_x \Pi_i\on{X}} \Tr{\psi_y B_{b|i}\on{Y}}\;,
  \end{equation}
  where we folded the creation of a state $\rho'_i$ into the measurement $B_{b|i}\on{Y}$.
  Then:
  \begin{align}
    & \qquad \sum_{xy} \payoff(1|x,y) p_{\EB}(1|x,y) \nonumber \\
    & = \sum_{xyi} \omega_{xy} \Tr{(\xi_x \Pi_i\on{X})^\top} \Tr{(\psi_y B_{1|i}\on{Y})^\top} \nonumber \\
    & = \sum_i \Tr{(\Pi_i\on{X} \otimes B_{1|i}\on{Y})^\top W} \le 0 \;,
  \end{align}
  as the element traced with $W$ is a product of positive semidefinite operators.

\section{Sparse witness decompositions}
\label{app:sparsewitness}
The construction given in Section~\ref{sec:construction} rests on the decomposition~\eqref{eq:witnessdecomposition} of an entanglement witness $W$ on pairs of input states $(\xi_x, \psi_y)$ with coefficients $\omega_{xy}$.
Notice that the statistics $p(b|xy)$ need only to be collected for the input pairs $(x,y)$ which have $\omega_{xy} \ne 0$.
We now provide a method to minimize the number of those input pairs.
\begin{proposition}
  Let $W$ be a Hermitian operator acting on $\hilbert{X} \otimes \hilbert{Y}$.
  Let $d=\min(\dim\hilbert{X}, \dim\hilbert{Y})$.
  Shifting the transpose to the left-hand side of~\eqref{eq:witnessdecomposition}, the operator $W^\top$ has a decomposition
  \begin{equation}
    \label{eq:witnessdecompositionappendix}
    W^\top = \sum_{xy} \omega_{xy} (\xi_x \otimes \psi_y)\;,
  \end{equation}
  where the number of nonzero $\omega_{xy}$ is at most $d^2 + 3$.
\end{proposition}
\begin{proof}
The operator $W^\top$ has the operator-Schmidt decomposition~\cite{Guhne2009}:
\begin{equation}
  W^\top = \sum_{i=1}^n \gamma_i (A_i \otimes B_i),
\end{equation}
where $n=d^2$, $\gamma_i \ge 0$, and $\{A_i\}$ and $\{B_i\}$ are Hermitian operators acting on $\hilbert{X}$ and $\hilbert{Y}$ respectively.
This form does not match with~\eqref{eq:witnessdecompositionappendix} as the $A_i$, $B_i$ do not necessarily represent density matrices: their eigenvalues can be negative and their trace can be different from one.
The real mismatch is given by nonnegative eigenvalues, as normalization is easily fixed by rescaling.
For simplicity, we first decompose $W^\top$ over unnormalized density matrices $\{\tilde{\xi}_x\}$ and $\{\tilde{\psi}_y\}$.
We define $\tilde{\xi}_0 = \openone$ and $\tilde{\psi}_0 = \openone$ with suitable dimension,
and write, for $x,y = 1,\dots, n$:
\begin{equation*}
  \tilde{\xi}_x = A_x + a_x \tilde{\xi}_0, \quad \tilde{\psi}_y = B_y + b_y \tilde{\psi}_0,
\end{equation*}
for minimal $a_x, b_y \ge 0$, such that all $\tilde{\xi}_x$, $\tilde{\psi}_y$ have nonnegative eigenvalues.
Then
\begin{equation*}
  W^\top = \mu ~ \tilde{\xi}_0 \otimes \tilde{\psi}_0 + \tilde{A} \otimes \tilde{\psi}_0 + \tilde{\xi}_0 \otimes \tilde{B} + \sum_{i=1}^n \gamma_i ~ \tilde{\xi}_i \otimes \tilde{\psi}_i
\end{equation*}
where
\begin{equation*}
  \mu = \sum_i \gamma_i a_i b_i, ~ \tilde{A} = -\sum_i \gamma_i b_i \tilde{\xi}_i, ~ \tilde{B} = -\sum_i \gamma_i a_i \tilde{\psi}_i
\end{equation*}
which is nearly of the required form apart from $\tilde{A}$ and $\tilde{B}$.
We again write
\begin{equation*}
  \tilde{\xi}_{n+1} = \tilde{A} + a_{n+1} \tilde{\xi}_0, \quad \tilde{\psi}_{n+1} = \tilde{B} + b_{n+1} \tilde{\psi}_0
\end{equation*}
for minimal $a_{n+1}, b_{n+1} \ge 0$ to ensure semidefinite positiveness.
This provides a final decomposition $W^\top = \sum_{x,y=0}^{n+1} \tilde{\omega}_{xy} ( \tilde{\xi}_x \otimes \tilde{\psi}_y )$ with nonzero coefficients:
\begin{equation*}
  \tilde{\omega}_{00} = \mu -a_{n+1} -b_{n+1}, \quad \tilde{\omega}_{0,n+1} = \tilde{\omega}_{n+1,0} = \tilde{\omega}_{ii} = 1
\end{equation*}
for $i=1,\ldots,n$.
The correct decomposition~\eqref{eq:witnessdecompositionappendix} is found by normalizing the states $\xi_x = \tilde{\xi}_x / \Tr{\tilde{\xi}_x}$, $\psi_y = \tilde{\psi}_y / \Tr{\tilde{\psi}_y}$, rescaling the coefficients $\tilde{\omega}_{xy} \to \omega_{xy}$ in the process.
The final decomposition has at most $d^2 + 3$ nonzero coefficients, to compare with the full tomographic data that corresponds to $(\dim\hilbert{X})^2 (\dim\hilbert{y})^2 \ge d^4$ coefficients.
\end{proof}

For the entanglement witness $W = \Phi_+ - \openone/2$ of our qubit depolarizing channel example, we have the operator-Schmidt decomposition
\begin{equation}
W^\top = (-\openone_4 + \sigma_x \otimes \sigma_x - \sigma_y \otimes \sigma_y + \sigma_z \otimes \sigma_z)/4\;.
\end{equation}
By following the process described above, we obtain a decomposition with only six nonzero coefficients.
The input states are
\begin{equation}
  \xi_0 = \frac{\openone}{2}, ~ \xi_1 = \frac{\openone + \sigma_x}{2}, ~ \xi_2 = \frac{\openone + \sigma_y}{2}, ~ \xi_3 = \frac{\openone + \sigma_z}{2}
\end{equation}
and
\begin{equation}
  \xi_4 = \frac{\openone + (\sigma_x - \sigma_y + \sigma_z)/\sqrt{3}}{2}
\end{equation}
where $\psi_i = \xi_i$.
The six nonzero coefficients are $\omega_{00} = 2(\sqrt{3}-1)$, $\omega_{04} = \omega_{40} = -\sqrt{3}$ and $\omega_{11} = -\omega_{22} = \omega_{33} = 1$.
Note the similarity of this decomposition with the one used in the experimental work~\cite{Xu2014}; however our construction is general and not tailored to a specific family of witnesses.

\end{document}